\documentclass{llncs}

\usepackage{lineno,hyperref}
\usepackage{xspace}
\usepackage{amsmath, amssymb,algorithm, algorithmicx,latexsym,amsfonts,amstext} 
\usepackage{algpseudocode}
\usepackage{listings}
\usepackage{xcolor,colortbl}
\usepackage{graphicx}
\usepackage{wrapfig}

\usepackage{float}
\usepackage[caption=false,font=footnotesize]{subfig}

\def \X {\ensuremath{\mathcal{X}}\xspace}
\def \U {\ensuremath{\mathcal{U}}\xspace}
\def \D {\ensuremath{\mathcal{D}}\xspace}

\def \G {\ensuremath{\mathcal{G}}\xspace}
\def \I {\ensuremath{\mathcal{I}}\xspace}

\def \L {\ensuremath{\mathcal{L}}\xspace}
\def \M {\ensuremath{\mathcal{M}}\xspace}
\def \C {\ensuremath{\mathcal{C}}\xspace}
\def \R {\ensuremath{\mathcal{R}}\xspace}

\definecolor{Gray}{gray}{0.95}
\definecolor{LightCyan}{rgb}{0.88,1,1}

\newcommand{\Reals}[0]{\ensuremath{\mathbb{R}}}

\begin{document}

\title{Parallel Reachability Analysis for Hybrid Systems}
\author{Undisclosed Authors}
\institute{Undisclosed Institutes}
\author{ Amit Gurung\inst{1} \and Arup Kumar Deka\inst{1} \and \newline
Ezio Bartocci\inst{2} \and Sergiy Bogomolov\inst{3}, Radu Grosu\inst{2} \and Rajarshi Ray\inst{1}\thanks{Corresponding author e-mail address: \email{rajarshi.ray@nitm.ac.in}}}

\institute{National Institute of Technology Meghalaya, India
\and Vienna University of Technology, Austria
\and Institute of Science and Technology Austria, Austria
}
\footnotetext{\noindent \textbf{Acknowledgements.} 
\small{This work was supported in part by DST-SERB, GoI under Project No. YSS/2014/000623 and by the European Research Council (ERC)
under grant 267989 (QUAREM) and by the Austrian Science Fund (FWF)
under grants S11402-N23, S11405-N23 and S11412-N23 (RiSE/SHiNE)
and Z211-N23 (Wittgenstein Award).} }
\maketitle

\begin{abstract}
 We propose two parallel state-space exploration algorithms for hybrid systems with the goal of enhancing performance on multi-core shared memory systems. The first is an adaption of the parallel breadth first search in the SPIN model checker. We show that the adapted algorithm does not provide the desired load balancing for many hybrid systems benchmarks. The second is a task parallel algorithm based on cheaply precomputing cost of post (continuous and discrete) operations for effective load balancing. We illustrate the task parallel algorithm and the cost precomputation of post operators on a support-function-based algorithm for state-space exploration. The performance comparison of the two algorithms displays a better CPU utilization/load-balancing of the second over the first, except for certain cases. The algorithms are implemented in the model checker XSpeed and our experiments show a maximum speed-up of $900\times$ on a navigation benchmark with respect to SpaceEx LGG scenario, comparing on the basis of equal number of post operations evaluated.
 
\end{abstract}



\section{Introduction}

Hybrid systems are a popular formal framework to model and 
verify safety properties in biological~\cite{Bartocci2009,Bartocci2016} and cyber-physical 
systems~\cite{Lee2015}. This formalism combines the classical discrete
state-based representation of finite automata with a continuous dynamics 
semantics in each state (or mode). 
A  hybrid  system  is  called  \emph{safe}  if  a
given set of \emph{bad states} are not reachable from a set of initial states.  
Hence, safety can be proved by performing a reachability analysis
that is in general undecidable~\cite{Henzinger95} for hybrid systems.
SpaceEx~\cite{FLGDCRLRGDM11,Bogomolov2015a,DBLP:conf/hvc/BogomolovFGGPPS14,Bogomolov2015b} is one of the most popular reachability-analysis tools for hybrid systems using semi-decision  procedures  based  on  over-approximation
techniques~\cite{GirardLG08,DBLP:conf/hybrid/Girard05}. The reachable states are represented as a collection of 
continuous sets, each one symbolically  represented.
The main two challenges to address with such set-based
methods are precision and scalability. Recently, algorithms using convex sets represented 
using support functions~\cite{GirardLG08,DBLP:conf/cav/GuernicG09} and zonotopes~\cite{DBLP:conf/hybrid/Girard05} have shown
a very good scalability.  However, all these techniques 
 were originally conceived to run sequentially.  
 Hence, they are currently not suitable to exploit the modern and powerful 
 multi-core architectures that would enable them to improve further their scalability.  

\paragraph{Our contribution}
 In this work, we provide two parallel algorithms for reachability analysis of hybrid systems
 that can leverage CPU  multi-core architectures 
 to  speed-up the  performance of the current technology.
 Our approach relies on the adaptation of 
 Holzmann's lock-free parallel breadth first exploration algorithm~\cite{parSpinHolzmann12}  
 recently implemented in the SPIN model checker.
We first extend the original algorithm  
 to deal with the symbolic reachable states and flowpipe
 computations that are the necessary ingredients of reachability analysis of hybrid systems. 
However, we notice that this first approach often results not ideal concerning the   
 CPU's cores utilization and the load balancing.
 This happens  when the number of symbolic states to be explored 
 is less than the number of processor's cores or when the cost
 of flowpipe computation varies drastically for different 
 reachable states. For this reason, we provide a second algorithm 
 that improves considerably the load balancing by efficiently
 precomputing the cost of certain operations.
 The two algorithms are implemented in XSpeed~\cite{DBLP:conf/hvc/RayGDBBG15} tool
 that is now able to handle also SpaceEx models using 
  the Hyst~\cite{DBLP:conf/hybrid/BakBJ15} model transformation 
  and translation tool for hybrid automaton models. Our experiments
  show a speed-up of up to $900\times$ on a navigation benchmark instance
  with respect to SpaceEx LGG scenario, comparing on the basis of equal number
	of post operations evaluated. The tool and the benchmarks reported in the paper 
	can be downloaded from \url{http://nitmeghalaya.in/nitm_web/fp/cse_dept/XSpeed/index_new.html}
\vspace{-2ex}
\paragraph{Related Work} In the last decade, there
has been an increasing interest in developing techniques for reachability 
analysis for hybrid systems. The tools currently~\cite{FLGDCRLRGDM11,DBLP:conf/cav/ChenAS13,KongGCC15,DBLP:conf/cade/PlatzerQ08} available 
can  perform reachability analysis of hybrid systems 
with linear dynamics and more than 200 variables~\cite{FLGDCRLRGDM11,DBLP:conf/cav/GuernicG09}. 
However, all these works are not suitable to exploit the powerful modern parallel multi-core architectures.
In our previous paper~\cite{DBLP:conf/hybrid/RayG15,DBLP:conf/hvc/RayGDBBG15} we addressed the problem of computing reachability 
analysis of continuous systems, but we did not handle hybrid systems.
Hence, to the best of our knowledge, this paper represents the 
first attempt to provide a parallel reachability analysis algorithm 
 for hybrid systems.

\vspace{-2ex}
\paragraph{Paper organization}
The rest of the paper is organized as follows. Section \ref{sec:prelims} 
provides the necessary background on hybrid automata and reachability analysis.
In Section \ref{sec:GJH-PBFS}, we show how to extend the Holzmann's lock-free parallel breadth first exploration algorithm~\cite{parSpinHolzmann12}  
to handle the state space exploration in hybrid automata. 
Section~\ref{sec:taskParallel}  addresses 
the load balancing problem introducing the notion of 
a task parallel algorithm. 
In Section \ref{sec:taskParallelSup} we
provide a concrete example of a task parallel algorithm 
in the context of support-function-based reachability analysis. 
Section \ref{sec:experiments} reports the experimental results to 
illustrate performance speed-up and CPU's core utilization. 
We conclude in Section~\ref{sec:conclusion}.

\section{Preliminaries}\label{sec:prelims}

A hybrid automaton is a mathematical model of systems exhibiting both continuous and discrete dynamics. A \emph{state} of a hybrid automaton is an n-tuple $(x_1,x_2,\ldots,x_n)$ representing the values of the $n$ continuous variables in an n dimensional system at an instance of time.

\begin{definition} A hybrid automaton is a tuple $(\L, \X, Inv, Flow, Init, \delta, \G, \M)$ where:
\begin{itemize}
	\item[-] $\L$ is the set of locations of the hybrid automata.
	\item[-] $\X$ is the set of continuous variables of the hybrid automata.
	\item[-] $Inv: \L \to \Reals^n$ maps every location of the automata to a subset of $\Reals^n$, called the invariant of the location. An invariant of a location defines a bound on the states within the location of the automata.
	\item[-] $Flow$ is a mapping of the locations of the automata to ODE equations of the form $\dot{x} = f(x), x \in \X$, called the flow equation of the location. A flow equation defines the evolution of the system variables within a location.
	\item[-] $Init$ is a tuple $(\ell_{init}, \C_{init})$ such that $\ell_{init} \in \L$ and $\C_{init} \subseteq Inv(\ell_{init})$. It defines the set of initial states of the automata.
	\item[-] $\G \subseteq \Reals^n$ is the set of guard sets of the automata.
	\item[-] $\M: \Reals^n \to \Reals^n$ is the set of assignment maps of the automata.
	\item[-] $\delta \subseteq \L \times \G \times \M \times \L$ is the set of transitions of the automata. A transition from a source location to a destination location in $\L$ may trigger when the \emph{state} $s$ of the hybrid automata lies in the guard set from $\G$. The map $\M$ of the transition transforms the state $s$ in the source location to a new state $s'$ in the destination location.	
\end{itemize} 
\end{definition}

A reachable state is a state attainable at any time instant $0 \leq t \leq T$ under its flow and transition dynamics starting from an initial state in $Init$. The flow dynamics evolves a state $(\ell,x)$ to another state $(\ell,y) $ in a location $\ell$ after $T$ time units such that $flow(x,T)=y$ and $flow(x,t) \in \textit{ Inv(} \ell$) for all  $t \in [0,T]$, where $flow$ is the solution to the flow equation of $Flow(\ell)$. Reachability analysis tools produce a conservative approximation of the reachable states of the automaton over a time horizon. Reachable states can be expressed as a union of \emph{symbolic states}. A symbolic state is a tuple \{$loc$, $\C$\} such that $loc \in \L$ and $\C \subseteq Inv(loc)$.

In Algorithm \ref{algo:reach} \cite{FLGDCRLRGDM11}, we show a generic reachability algorithm for hybrid automata. The algorithm maintains two data structures, $Wlist$ and $R$. $Wlist$ stores the list of symbolic states to initiate the exploration of reachable states and $R$ stores the already visited reachable states. $Init$ is a symbolic state depicting the initial states given as an input. $Wlist$ and $R$ are initialized to $Init$ and $\emptyset$ respectively in line \ref{listInit}. A symbolic state $S$ is removed from the $Wlist$ at each iteration and explored for reachable states. The operators $PostC$ and $PostD$ are applied consecutively. $PostC$ takes a symbolic state as argument and computes the reachable states under the continuous dynamics of the location. The result of $PostC$ is a symbolic state, say $R'$. When $R'$ is contained in $R$, there is no newly explored state and the next symbolic state in $Wlist$ is explored as shown in line \ref{goto}. On the contrary, when new states are found in $R'$ not in $R$, they are included in $R$ in line \ref{union} . The operator $PostD$ takes a symbolic state and returns a set of symbolic states obtained under the discrete dynamics. This new set of symbolic states, shown as $R''$, is added to the $Wlist$ for further exploration.

\begin{algorithm}[!htb]
\caption{Reachability Algorithm for Hybrid Automata}
\begin{algorithmic}[1] 
\Procedure{Reach-ha}{ha,$Init$}
\State $Wlist \gets Init$; $R \gets \emptyset$; \label{listInit}
\While {Wlist $\not=\emptyset$}
	\State $S \gets Wlist.pop()$; $R' \gets PostC(S)$
	\If {$R' \subseteq R$} go to step 3 \label{goto}
	\Else  \textit{ } $R \gets R \cup R'$ \label{union}
	\EndIf
	\State $R'' \gets PostD(R')$; $Wlist.add(R'')$ 
\EndWhile
\EndProcedure
\end{algorithmic}
\label{algo:reach}
\end{algorithm}

\section{Parallel Breadth First Search}\label{sec:GJH-PBFS}

It is worthy noting that the $PostC$ computation for symbolic states in $Wlist$ is independent of each other and therefore can be potentially parallelized. In this work, we exploit this inherent parallelism and propose parallel breadth first search (BFS) algorithm. In a multi-threaded implementation, threads can compute $PostC$ and $PostD$ operations in parallel, however the access to the shared data structure $Wlist$ and $R$ has to be mutually exclusive to avoid race condition. Mutual exclusion can be accomplished by using locks or semaphores, however at the price of additional overhead. Moreover, such an implementation may not ensure an effective load balancing as illustrated later in the following.

In Algorithm \ref{algo:pbfs} we show how to avoid the overhead of the mutual exclusion discipline by adapting the parallel lock-free breadth first search algorithm proposed in \cite{parSpinHolzmann12} to hybrid system state-space exploration. Our adapted algorithm is referred as A-GJH (Adapted Gerard J. Holzmann's) algorithm in the paper. The algorithm uses two copies of $Wlist$, each being a two-dimensional list of symbolic states. At each iteration, symbolic states are read from the $Wlist[t]$ copy and new symbolic states are written to the $Wlist[1-t]$ copy. At the first iteration the value $t$ is 0 and at the end of each iteration (see line \ref{rwswitch}) of the main while loop it is assigned to $1-t$. In this way, in the next iteration the write list becomes the read list and vice-versa. There are $N$ threads, one thread per core, which computes the post operations in parallel. All the symbolic states present in the row $Wlist[t][w]$ are sequentially processed by the thread indexed by $w$, shown in line \ref{thrtask1}-\ref{thrtask2}. A symbolic state when undergoes the $PostD$ operation generates a list of successor symbolic states, to be processed in the next iteration. Each successor state is added to the list $Wlist[1-t][w'][w]$, where $w'$ is randomly selected between one to $N$ (line \ref{symadd}). This random distribution of new states across rows of $Wlist[1-t]$ is used for load balancing. Since each thread indexed at $w$ writes to the list at $Wlist[1-t][w'][w]$, there is no write contention in $Wlist[1-t]$. Checking containment of symbolic states in the result data structure $R$ is avoided in the algorithm since it is an expensive operation. The exploration is instead bounded by the number of levels of exploration of the automata. The symbolic states of $Wlist[t]$ are explored when all the $N$ threads terminate and synchronize at line \ref{thsync} ensuring a breadth-first exploration. The algorithm terminates when there are no successor states in $Wlist[1-t]$ for further exploration or when the breadth exploration reaches a certain bound.  

\begin{algorithm}[!htb]
\caption{Adapted G.J. Holzmann's Algorithm}
\begin{algorithmic}[1] 
\Procedure{Reach-PBFS}{ha, $Init$, bound}
\State $t = 0$, $level = 0, N = Cores$
\State $Wlist[2][N][N]$	\Comment{$Wlist$ is a read/write list of symbolic states.}
\State $Wlist[t][1][1] = Init$
\State Create $w = 1$ to $N$ threads \Comment{N worker threads}
\Repeat
	\For {$q=1$ to $N$} \label{thrtask1}
		\For{each s in q}
			\State delete s from $Wlist[t][w][q]$
			\State $R'[i] \gets PostC(s)$ 
			\State $R''[i] \gets PostD(R'[i])$ 
			\State $w'=$  choose random $1 \ldots N$
			\State add $s' \in R''[i]$ to $Wlist[1-t][w'][w]$ \label{symadd} \label{thrtask2}
		\EndFor
	\EndFor
\State Barrier synchronization \Comment{Thread synchronization to ensure BFS} \label{thsync}
\If{$Wlist[1-t][1 \ldots N][1 \ldots N]$ is all empty}
	\State done = true;
\Else
	\State $t = 1 - t$	\label{rwswitch} \Comment{Read/Write switching}
	\State $level \gets level + 1$
\EndIf
\Until{$!done \textit{ OR } level > bound$ }
\EndProcedure
\end{algorithmic}
\label{algo:pbfs}
\end{algorithm}

\begin{figure}[!htb] 
   \includegraphics[width=\textwidth]{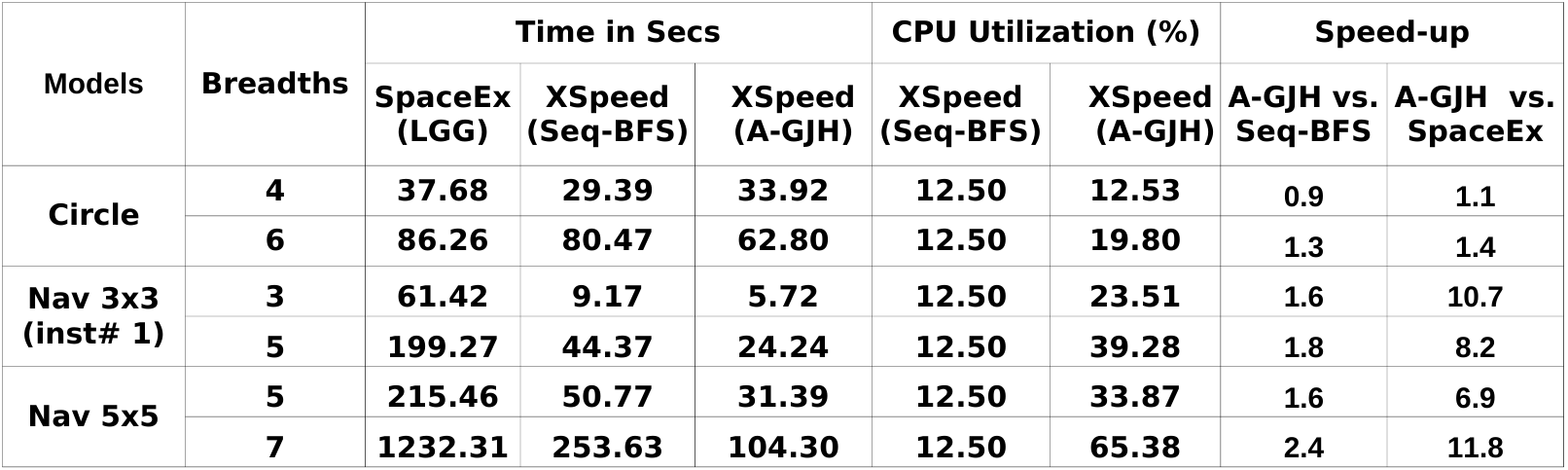}
\caption{Moderate CPU utilization and performance speed-up with A-GJH algorithm.\label{tab:GJHadopt}}
\label{fig:LBMotivation}
\end{figure}

 
\subsection{\label{sub:LoadBalancing}Load Balancing}
The clever use of the data structures in A-GJH algorithm provides freedom from locks and reasonable load balancing when there are sufficiently large number of symbolic states in the waiting list. However, the load balancing in A-GJH does not perform well when the number of symbolic states in the waiting list is less than the number of cores of the processor. Since the symbolic states are distributed randomly to the $N$ cores in line \ref{thrtask2} for exploration, some of the cores remain idle when there are not enough states to be explored. For this reason there are benchmarks for hybrid systems reachability analysis where an incorrect load balancing results in a low time utilization of the available cores. For example, Figure \ref{fig:LBMotivation} shows that while the A-GJH running in a 4 core processor provides some improvements in performance compared to SpaceEX LGG (Le Guernic-Girard) and the sequential BFS, the CPU core utilization remain very modest. The best utilization is $65\%$ in the NAV 5X5 benchmark for 7 levels exploration and the worst is $12\%$ in the Circle model. In the Circle model, there are only 2 symbolic states for exploration at each BFS iteration, keeping most of the CPU cores idle.

Another principle problem in load balancing is that the cost of flowpipe computation may vary drastically for different symbolic states. This is illustrated on a $3 \times 3$ Navigation benchmark in Figure \ref{fig:ldUndersubs} \cite{FehnkerI04}. The benchmark models the motion of an object in a 2D plane partitioned as a $3 \times 3$ grid. Each cell in the grid has a width and height of 1 unit and have a desired velocity $v_d$. In Figure \ref{fig:ldUndersubs} the cells are numbered from 1 to 9 and the respective desired velocities are shown with a directed vector. Note that there is no desired velocity shown in cell 3 and 7 to distinguish them from the others. Cell 3 is the target whereas cell 7 is  the unsafe region. The actual velocity of the object in a cell is given by the differential equation $\dot{v} = A(v - v_d)$, where $A$ is a $2\times2$ matrix. There is an instantaneous change of dynamics on crossing over to an adjacent cell. The green box is an initial set where the object can start with an initial velocity. The red region shows the reachable states under the hybrid dynamics after a finite number of cell transitions.

\begin{figure}[!htb]	
	\subfloat[Flowpipe after exploring 2 levels having 4 new symbolic states\label{fig:a}]
		{\centering{}\includegraphics[scale=0.35]{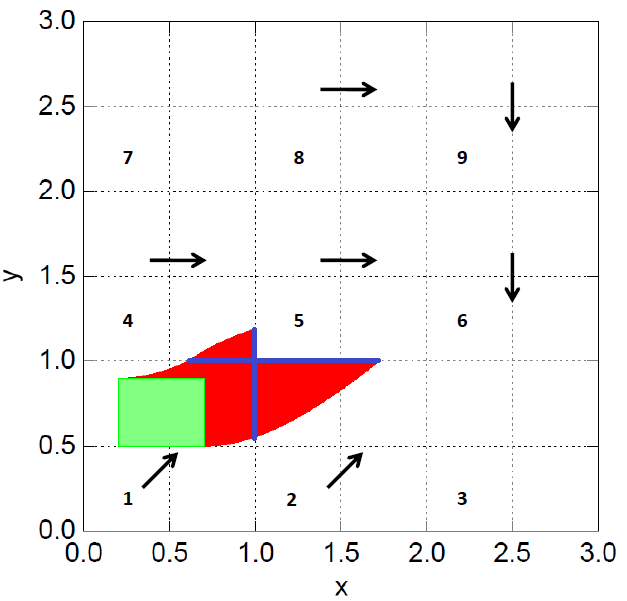}} \quad{}
	\subfloat[Flowpipe after exploring 3 levels\label{fig:b}]
		{\centering{}\includegraphics[scale=0.35]{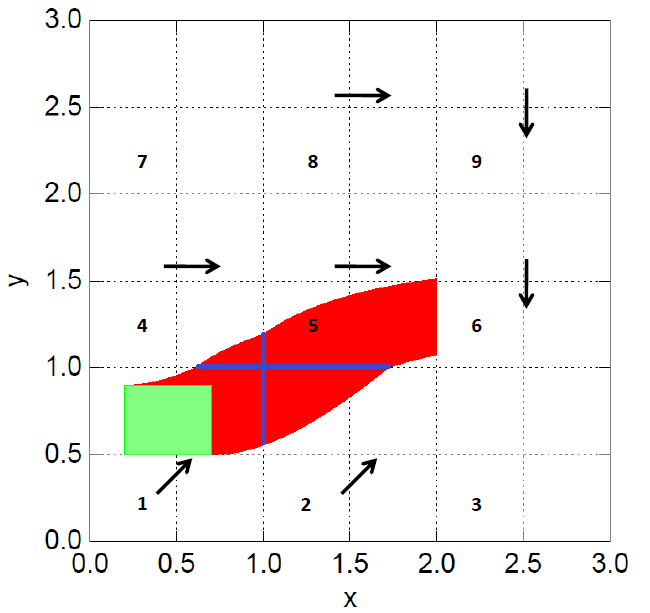}}
\caption{Illustrating Load Balancing Problem with Flowpipes of Varying Cost\label{fig:ldUndersubs}}
\end{figure}

Figure \ref{fig:ldUndersubs} shows the reachable states after two and three levels of exploration in Algorithm \ref{algo:pbfs}. There are four symbolic states, $S_1$, $S_2$, $S_3$ and $S_4$, waiting to be explored after two levels of bfs, shown in blue. The symbolic states $S_1$, $S_2$ are \{$1$, $B_1$\}, \{$1$, $B_2$\} and $S_3$, $S_4$ are \{$5$, $B_3$\} and \{$5$, $B_4$\} respectively, where $1$ and $5$ are the location ids and $B_1$, $B_2$, $B_3$ and $B_4$ are the blue regions in the boundary of location with ids 1 and 4, 1 and 2, 4 and 5, 2 and 5 respectively. Algorithm \ref{algo:pbfs} spawns four threads, one each to compute the flowpipe from the symbolic states. In a four core processor, this seems an ideal load division. However, observe in Figure \ref{fig:b} that out of the four flowpipes, the two from $S_1$ and $S_2$ do not lead to new states since they start from the boundary of location id 1 and the dynamics pushes the reachable states outside the location invariant. This implies that the flowpipe computation cost for $S_1$ and $S_2$ are low and the two cores assigned to these flowpipe computation finish early and waits at the synchronization point until the remaining two busy cores complete. Such a situation keeps the available cores under-utilized due to the idle waiting of $50\%$ of cores. Similarly, the cost of $PostD$ operation also varies causing idle waiting of threads assigned to low cost computations.

\section{Task Parallel Algorithm}\label{sec:taskParallel}
In the following we propose an alternative approach to improve load balancing based on computational cost of the tasks encountered during the exploration. The idea is to identify the \emph{atomic tasks} in a flowpipe ($PostC$) computation. The atomic tasks across all flowpipe computations in the breadth search is a measure of the total work-load, at the present breadth of the exploration. For an effective load balancing, this work-load is distributed evenly between the cores of the processor. Similar tasks distribution can be applied also to the computation of discrete transitions ($PostD$). Algorithm \ref{algo:task-pbfs} shows this approach.

\begin{algorithm}[!htb]
\caption{Task Parallel Breadth First Exploration}
\begin{algorithmic}[1] 
\Procedure{Reach-Task-PBFS}{ha, $Init$, bound}
\State $t = 0$, $level = 0$, $N = 1$, $CostC=0$, $CostD=0$
\State $Wlist[2][N]$	\Comment{$Wlist$ is a read/write list of symbolic states.}
\State $Wlist[t][0] = Init$

\Repeat
	\For{each s in $Wlist[t]$}
		\State $CostC = CostC + flow\_cost(s)$ \label{flowCost1}
		\State Break $PostC(s)$ into atomic tasks and add to Tasks list \label{divide1}
	\EndFor
	\State $Tasks\_Per\_Core$ = $\lceil CostC/\#Cores \rceil$ \Comment{Even distribution of tasks to cores}
	\For{Threads with id $w = 1$ to $N$} \Comment{N worker threads}\label{parTask}
		\State Execute Tasks\_Per\_Core exclusive tasks from the Tasks list \label{parTask1}
		\State Add results to Res
	\EndFor
	\State Barrier Synchronization \label{parTaskFin} 
	\For{each s in $Wlist[t]$, create thread indexed w} \Comment{w worker threads}
		\State $flow[w] = Res.join()$ \Comment{Combine task results to get flowpipe} \label{join1}  
	\EndFor
	\State Barrier Synchronization
	\For{each s in $Wlist[t]$}
		\State $CostD = CostD + jump\_cost(flow[s])$ \label{flowCost}
		\State Break $PostD(s)$ into atomic tasks and add to Tasks list \label{divide}
	\EndFor
	\State $Tasks\_Per\_Core$ = $\lceil CostD/\#Cores \rceil$ \Comment{Even distribution of tasks to cores}
	\For{Threads with id $w = 1$ to $N$} \Comment{N worker threads}\label{parTask}
		\State Execute Tasks\_Per\_Core exclusive tasks from the Tasks list
		\State Add results to Res
	\EndFor
	\For{each s in $Wlist[t]$, create thread indexed w} \Comment{w worker threads}
	\State $R'[w] = Res.join()$ \Comment{Combine task results to get succ symbolic state} \label{join2}
	\State add each s in $R'[w]$ to $Wlist[1-t][w]$ \label{listWrite}
	\EndFor
	\State Barrier Synchronization			
	
	\If{$Wlist[1-t]$ is empty} \textit{ done = true}
	\Else
		\State $t = 1 - t$	\Comment{Read/Write switching}
		\State $N$ = sum of size of all lists in $Wlist[t]$ 
		\State Resize Wlist[1-t][N], $level = level + 1$, $CostC = 0$, $CostD = 0$
	\EndIf
\Until{$!done \textit{ OR } level > bound$ }
\EndProcedure
\end{algorithmic}
\label{algo:task-pbfs}
\end{algorithm}

In particular, the instruction in line \ref{flowCost1} obtains an estimated cost of computing a flowpipe from a symbolic state using the function $flow\_cost$. After the flowpipe costs for all symbolic states in $Wlist$ are computed, line \ref{divide1} breaks the flowpipe computations into atomic tasks and adds them into a tasks list. In line \ref{parTask1} the atomic tasks are evenly assigned to the processors cores. In line \ref{join1} the results of the atomic tasks computed in parallel are then joined together to obtain a flowpipe. Similar load division is carried out for $PostD$ operation. The successor states obtained from each flowpipe are added in the write list $Wlist[1-t][w]$, by a thread indexed at $w$ in line \ref{listWrite}. The exclusive access of the threads to the columns of $Wlist[1-t]$ eliminate the write contention. 

A challenge in this approach is to devise an efficent method for computing the cost of flowpipe and discrete-jump computation for balanced load distribution. Efficient methods and data structures for splitting $PostC$, $PostD$ into atomic tasks and merging their results are very important in order to avoid that the extra required overhead would affect the gain obtained with the parallel exploration. In the next section, we propose some procedures to compute cost of post operations for a support-function algorithm.

\section{Task Parallellism in Support Function Algorithm} \label{sec:taskParallelSup}

A common technique in flowpipe computation consists in discretizing the time horizon $T$ into $N$ intervals with a chosen time-step $\tau=T/N$ and over-approximating the reachable states in each interval by a closed convex set, say $\Omega$. A flowpipe can be represented as a union of convex sets as shown in Equation \ref{eqn:ReachSet}.

\begin{equation}\label{eqn:ReachSet}
\begin{split}
	& Reach_{[0,T]}(\X_0) \subseteq \bigcup_{i=0}^{N-1} \Omega_i \\
	& Reach_{[i\tau,(i+1)\tau]}(\X_0) \subseteq \Omega_i \textit{, } \forall 0 \leq i < N \\
\end{split}
\end{equation}
 
The representation of the convex sets $\Omega$ is a key in the efficiency and scalability of reachability algorithms. Polytopes \cite{DBLP:conf/hybrid/ChutinanK99,Frehse08} and zonotopes\cite{DBLP:conf/hybrid/Girard05} are popular geometric objects suitable to represent convex sets. More recently, support functions \cite{GirardLG08,DBLP:conf/cav/GuernicG09} have been proposed as a functional representation of compact convex sets. SpaceEx \cite{FLGDCRLRGDM11} and XSpeed \cite{DBLP:conf/hvc/RayGDBBG15} tools implement support-function-based algorithms for flowpipe computation due to the promise it provides in scalability. We now present the preliminaries of support functions necessary to introduce the task parallelism in the algorithm.

\subsection{Support functions}
\begin{definition}\cite{RW98}
 Given a nonempty compact convex set $\X \subset \Reals^n$ the \emph{support function} of $\X$ is a function $sup_{\X}:\Reals^n \to \Reals$   defined as:
 \begin{equation}\label{def:sf}
  sup_{\X}(\ell) = \sup\{\ell \cdot x \mid x \in \X\}
 \end{equation}
\end{definition} 

\begin{definition}\label{eqn:tempPoly}
Given a support function $sup_{\X}$ of a compact convex set $\X$ and a finite set of template directions $\D$, A template polytope of the convex set $\X$ is defined as:
\begin{equation}
	Poly_\D (\X) = \bigcap_{l_i \in \D}{l_i.x \le sup_{\X}(\ell_i)}
\end{equation}
\end{definition}
  
The support-function algorithm in \cite{GirardLG08} proposes a flowpipe computation by computing the template polyhedral approximation of the convex sets $\Omega$ by sampling their support functions in the template directions. The algorithm is for linear dynamics with non-deterministic inputs of the form:

\begin{equation}
\dot{x} = Ax(t) + u(t), \qquad u(t) \in \U, x(0) \in \X_0 \label{eq:flowEquation}
\end{equation}
where $x(t),u(t) \in \Reals^n$, $A$ is a real-valued $n \times n$ matrix and $\X_0, \U \subseteq \Reals^n$ are the initial states and the set of inputs given as compact convex sets.

\subsection{Flowpipe Cost Computation}
We define the cost of computing a flowpipe by considering a support function sample as the logical atomic task in the computation.

\begin{definition} \label{def:flowCost}
Given a time horizon $T$, time discretization factor $N$, a set of template directions $\D$ and an initial symbolic state $s=(loc,\C)$, the cost of computing the flowpipe with postC(s) is given by:
\begin{equation}
flow\_cost(s) = j.|\D| 
\begin{cases}
& j = max \big\{ i \mid 0 \leq i \leq N \textit{, }\forall 0 \leq k \leq i, \Omega_k \vdash Inv(loc) \big\} \\
& \Omega_k \vdash Inv(loc) \textit{ if and only if } \Omega_k \cap Inv(loc) \neq \emptyset \\
\end{cases}
\end{equation}     
\end{definition}

The longest sequence $\Omega_0$ to $\Omega_j$ such that the convex sets satisfy the location invariant is identified. Since computing polyhedral approximation of the convex sets $\Omega$ requires sampling support function in each direction of the set of template directions $\D$, the $flow\_cost$ essentially gives us the number of support function samplings, i.e. the atomic tasks, that is to be completed in order to compute the flowpipe. To compute $flow\_cost$, it is necessary to find the longest sequence $\Omega_0$ to $\Omega_j$ satisfying the location invariant $Inv(loc)$. Assuming polyhedral invariants, checking the invariant satisfaction can be performed using the following proposition.

\begin{proposition}\cite{LeGuernic09} \label{invCheck}
Given a polyhedra $\I = \bigwedge_{i=1}^{m} \ell_{i}\cdot x \leq b_i$ and a convex set $\Omega$ represented by its support function $sup_{\Omega}$,
$\Omega \vdash \I$ if and only if $-sup_{\Omega}(-\ell_i) \leq b_i$, for all $1 \leq i \leq m$.
\end{proposition}

A procedure to identify the largest sequence is to apply Proposition \ref{invCheck} to each convex set starting from $\Omega_0$ iteratively  until we find a $\Omega_j$ such that $\Omega_j \nvdash Inv(loc)$. The time complexity of the procedure is $O(m\cdot N \cdot f)$, where $f$ is the time for sampling the support function, $m$ is the number of invariant constraints and $N$ is the time discretization factor. We propose a cheaper algorithm with fewer support functions samplings for a class of linear dynamics $\dot{x} = Ax(t) + u$, with $u$ being a fixed input. Fixed input leads to deterministic dynamics allowing to compute the reachable states symbolically at any time point .

\begin{proposition}\label{exactReach}
Given an initial set $\X_0$ and dynamics $\dot{x} = Ax(t) + u$ with $A$ being invertible, the set of states reachable at time $t$ is given by:
\begin{equation}
X(t)=e^{At}x_0 \oplus A^{-1}(e^{At}-I)(v)
\end{equation} 
\end{proposition}
 
The idea of the procedure shown in Algorithm \ref{algo:InvariantCheck}, is to use a coarse time-step to compute reachable states using Proposition \ref{exactReach} and detect an approximate time for crossing the invariant. Once the invariant crossing time is detected, similar search is followed by narrowing the time-step for a finer search near the boundary of the invariant for a desired precision. The procedure is illustrated on a toy model of a counter clockwise circular rotation dynamics as shown in Figure \ref{fig:CircleAutomaton}. The model has two locations with the same dynamics but different invariants. The transition assignment maps does not modify the variables. Figure \ref{fig:CircleTimeSearch} illustrates the procedure. The initial set on the location is shown in blue. The red sets are the reachable images of the initial set computed at coarse time steps to detect invariant crossing, followed by computing the images at finer time-steps shown in green near the invariant boundary for detecting an upper bound on the time of crossing the invariant with a desired precision. After computing this time, say $t'$, the $flow\_cost$ is obtained using Definition \ref{def:flowCost} with $j=t'/\tau$. However, the problem with the procedure is when it is possible for a reachable image to exit and re-enter the invariant within the chosen time-step. In such cases, the approximation error in the time returned by the procedure can be substantial. Constant dynamics and convex invariant $\I$ will not have such a scenario and the approximation error can be bounded.

\begin{theorem}
For a class of dynamics $\dot{x} = k$ , where $k$ is a constant, let $t$ be the exact time when reachable states from a given initial set $\X_0$ violate the convex location invariant $\I$. Let $\delta_C$ and $\delta_F$ be the coarse and fine time steps chosen to detect approximate time $t'$ of invariant violation. The approximation error $|t-t'| \leq \delta_F$.   
\end{theorem}
\begin{proof}
Constant dynamics have a fixed direction of dynamics and therefore, convexity property ensures that the reachable states cannot exit and re-enter $\I$. Reachable states $X(t) = X_0 \oplus kt$ is exactly represented using its support function. Algorithm \ref{algo:InvariantCheck} samples the support function at $\delta_F$ time-steps to identify the time instant $t'$ of crossing $\I$, which implies $|t-t'| \leq \delta_F$.      
\end{proof}

\begin{algorithm}[!htb]
\caption{Detecting time of invariant crossing with varying time-step}
\begin{algorithmic}[1] 
\Procedure{Invariant-Crossing Time Detection}{$\I$, $\X_0$, T}
	\State $discretization = 10$ , $\tau = T/discretization$ \Comment{Coarse Time-step}
	\State $i = 0$, $\R(0) = \X_0$
	\While{$\R(\tau . i) \vdash \I$}
		$i = i + 1$  \Comment{Widened Search}
	\EndWhile
	\If {$i > 1$} $t1 = \tau * (i-1)$
	\Else \textit{ return} 0 
	\EndIf 
	\State $\tau = \tau/discretization$, $i = 0$ \Comment{Fine Time-step}
	\While{$\R(t1+i*\tau) \vdash \I$} $i = i + 1$ \Comment{Narrowed Search}
	\EndWhile
	\State return $t1 + i*\tau$ \Comment{An upper bound on invariant crossing time}
\EndProcedure
\end{algorithmic}\label{algo:InvariantCheck}
\end{algorithm}


\begin{figure}[!htb]
\subfloat[Hybrid Automaton\label{fig:CircleAutomaton}]
	{\centering{}\includegraphics[width=5.3cm,height=2.7cm]{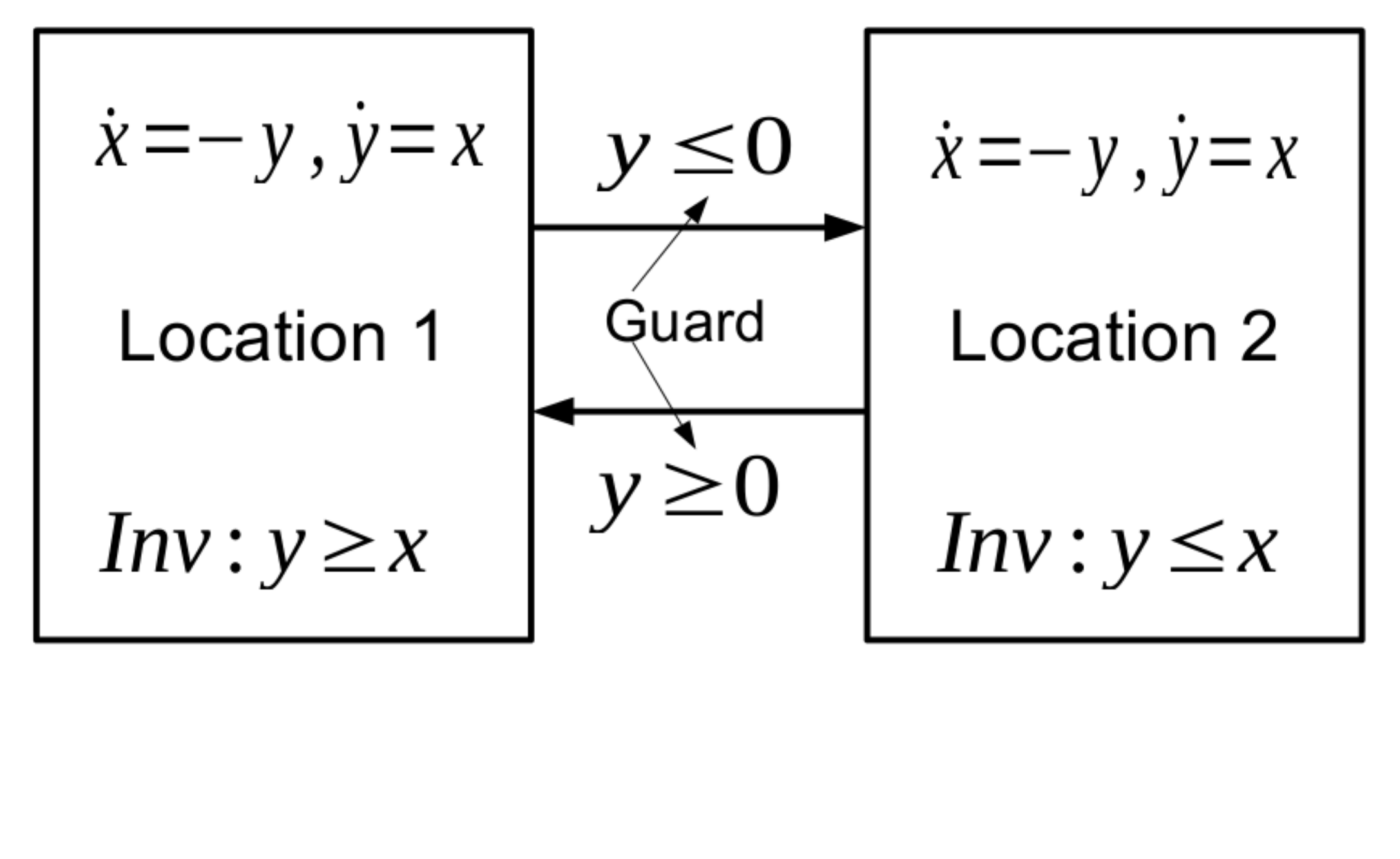}} \quad{}
\subfloat[Searching Time of Invariant Crossing with Widening-Narrowing Time-steps\label{fig:CircleTimeSearch}]
	{\centering{}\includegraphics[scale=0.35]{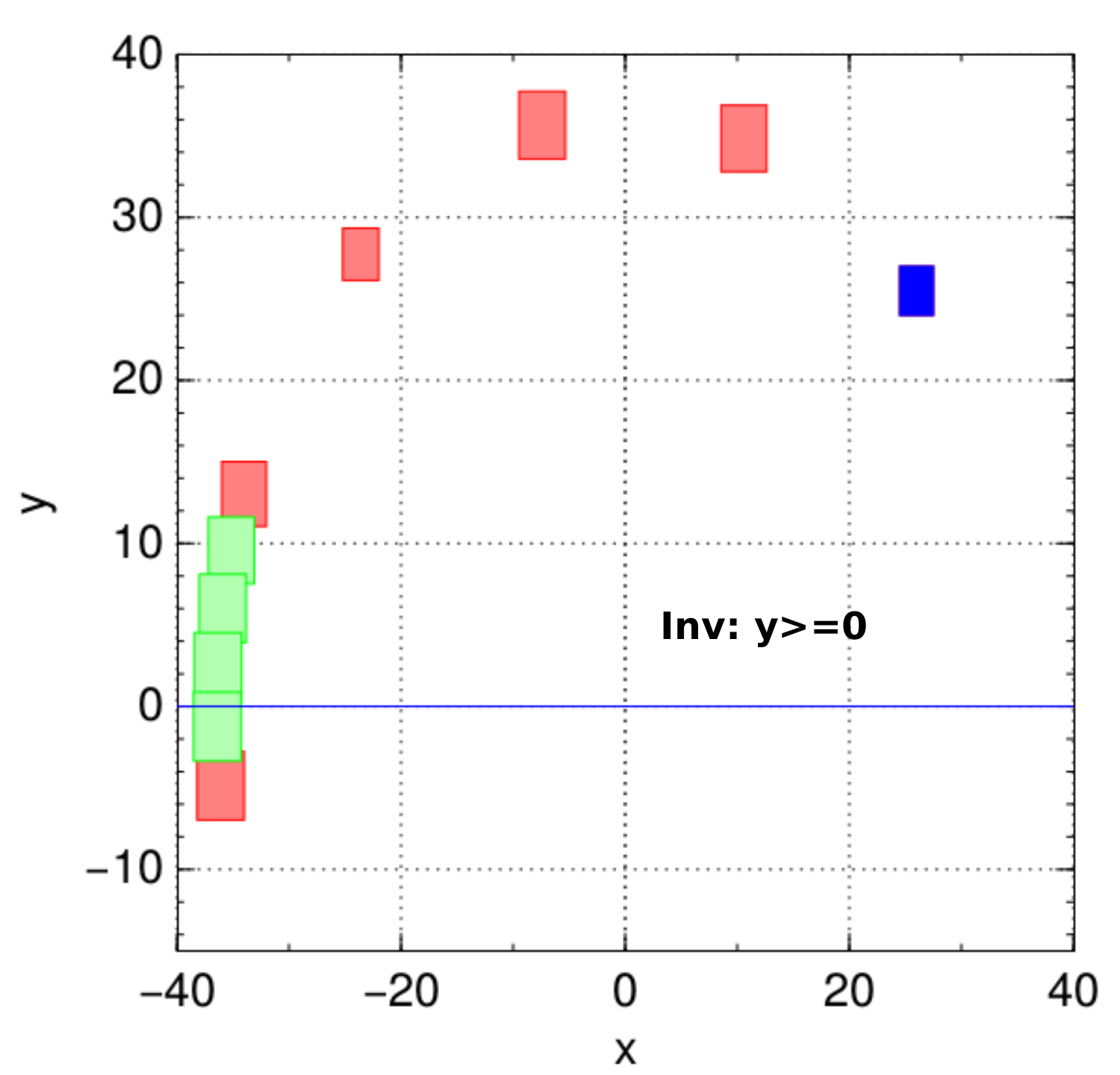}}
\caption{(a) A hybrid automaton of a toy example with circle dynamics. (b) Evaluation of the invariant crossing time detection algorithm on a circle model. The red sets are computed at coarse time-steps and the green sets computed at finer time-steps \label{fig:CircletimeSearch}}
\end{figure}

\subsection{Discrete-Jump Cost Computation}\label{Cost-postD}

The $PostD$ computation performs the flowpipe intersection with the guard set followed by image computation. Considering a flowpipe having sets $\Omega_0$ to $\Omega_j$, each of these sets are applied with intersection operation and a map for non-empty intersection. Assuming intersection and image computation as the atomic task, the cost of $PostD$ on a flowpipe $\cup_{i=0}^{j} \Omega_i$ will be $j$, which can be obtained from the $flow\_cost$ computation in Definition \ref{def:flowCost}. The addition of the cost of post operations for all symbolic states in the waiting list is used to uniformly distribute atomic tasks of post operations across the cores using multi-threading. Further details on the data-structures and task distribution is omitted due to the lack of space. The task parallel support-function-algorithm is referred as TP-BFS in the text that follows.

\section{Experiments}\label{sec:experiments}
The parallel algorithms are implemented with multi-threading using OpenMP compiler directives. Figure \ref{fig:experiments} shows the performance comparison between Reachability analysis with SpaceEx (LGG), Sequential BFS, A-GJH and TP-BFS. The benchmarks are: A two dimensional oscillator circuit model, model of a bouncing ball under gravity, model of a circular rotation dynamics, three instances of the Navigation benchmark with 9 ($3 \times 3$) locations and one instance each of 25 ($5 \times 5$) and 81 ($9 \times 9$) locations respectively. We conducted our experiments in a 4 core Intel i7-4770, 3.40GHz and 8GB RAM with hyper-threading enabled. The results are for a time horizon of 10 units, box template direction as parameters. The sampling time in Circle model is $1e-5$, in Oscillator, Timed Bouncing Ball, Navigation $3 \times 3$ and $5 \times 5$ instances are $1e-4$ and $0.1$ units for Navigation $9 \times 9$ instance respectively.
 
\begin{figure}[!htb]
\centering{}
\subfloat[XSpeed\label{fig:XSpeed_Nav09}]
	{\includegraphics[scale=0.3]{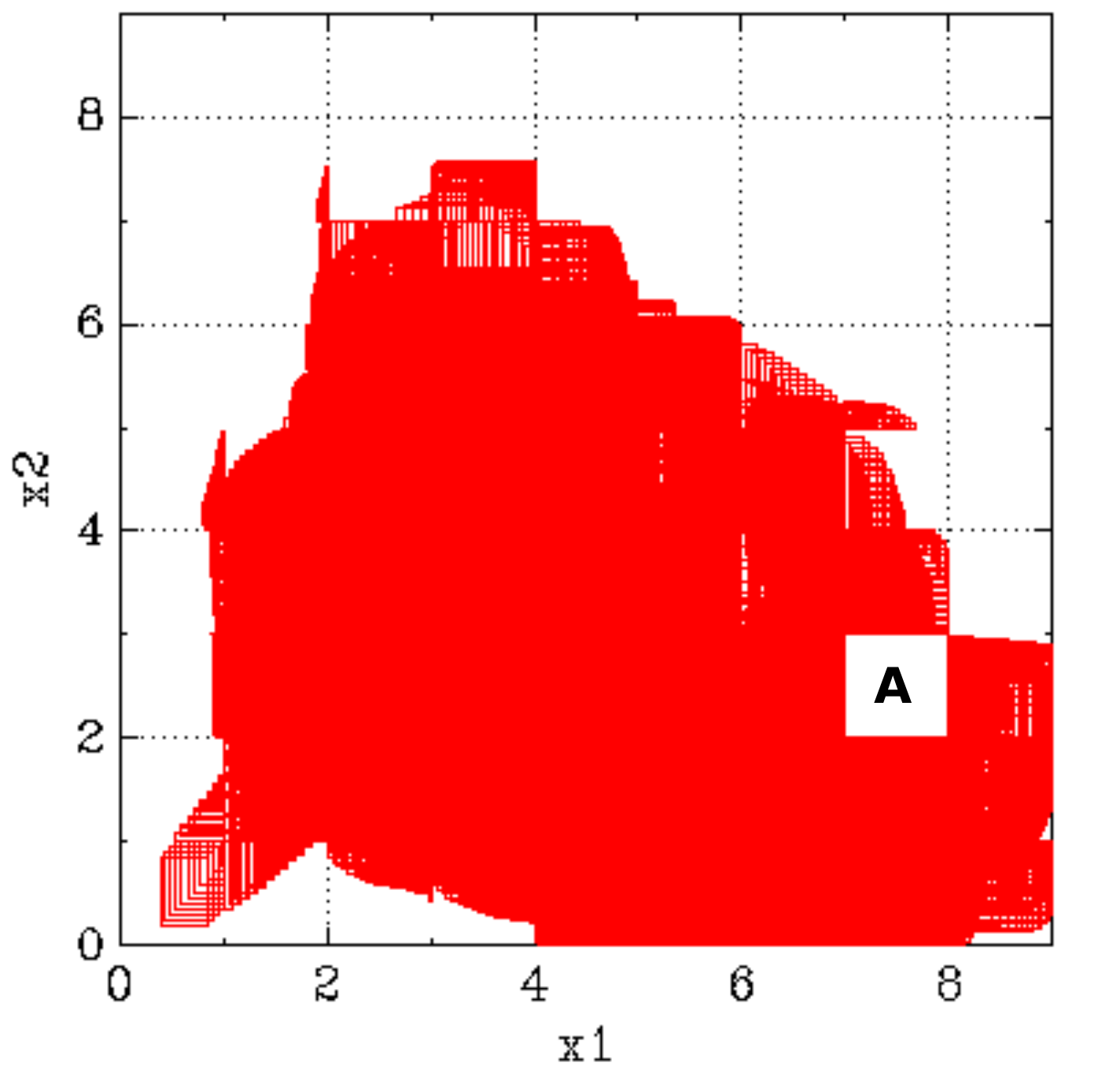}}
	 \quad{}
\subfloat[SpaceEx(LGG)\label{fig:SpaceEx_Nav09}]
	{\includegraphics[scale=0.3]{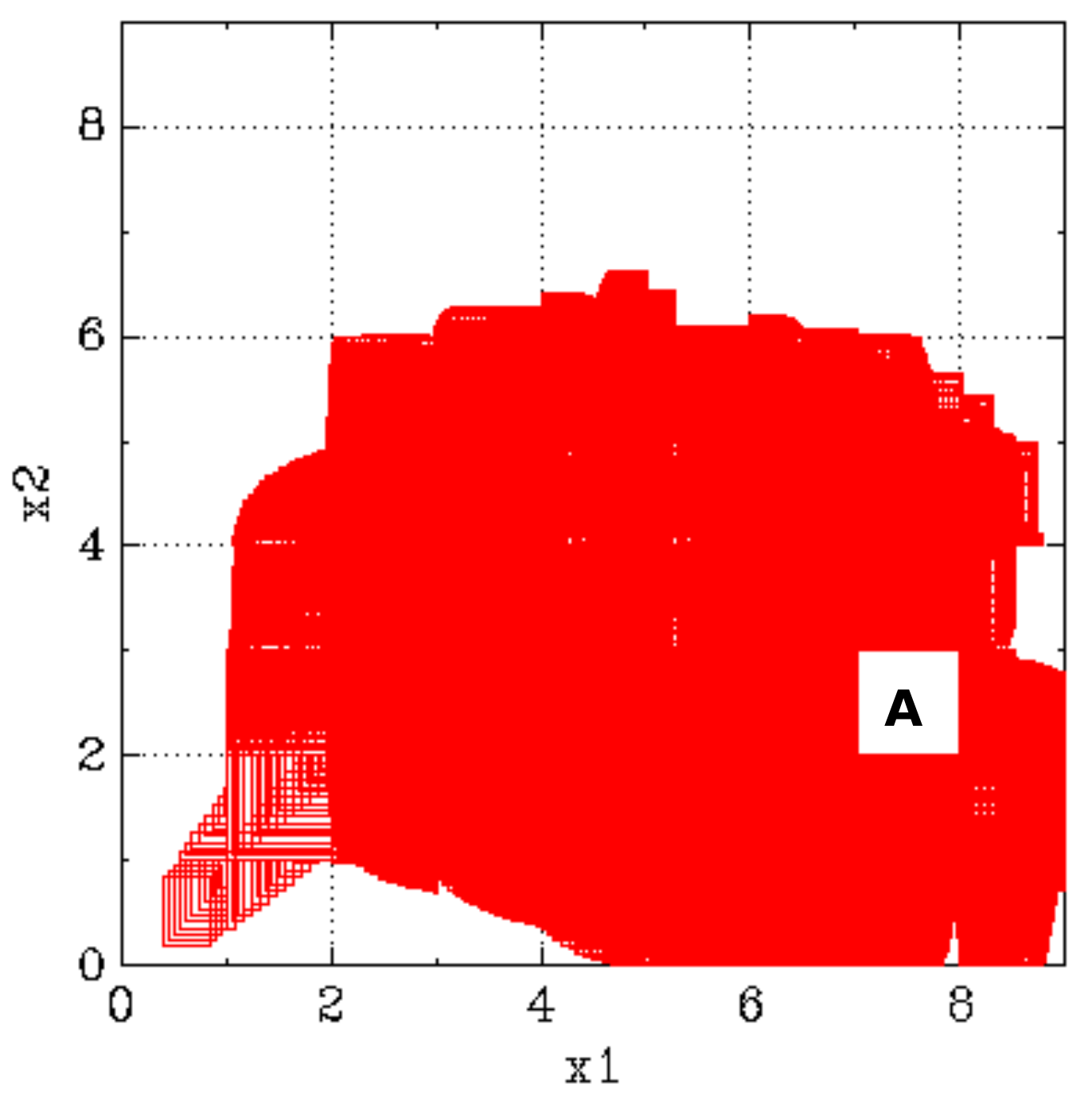}}
\caption{Reachable region of a $9 \times 9$ Navigation benchmark instance by (a) XSpeed after 13 BFS-levels (105563 post operations) and (b) SpaceEx's LGG algorithm (105563 post iterations) \label{fig:reachRegionNav09}
}
\end{figure}

In order to relate our results with SpaceEx (LGG) scenario, we apply the same number of post operations in both XSpeed and SpaceEx, with similar parameters. We count the number of post operations for a given bound on the BFS level in XSpeed and the exploration with SpaceEx is bounded with the same count on post operations (by setting the argument iter-max). However, note that the cost of post operations could be different for the symbolic states in XSpeed and in SpaceEx. Therefore, comparison on the number of post operations is not perfectly fair but we could not find a better means of comparing. Figure \ref{fig:reachRegionNav09} shows that the reachable region obtained from XSpeed on a Navigation benchmark after 13 BFS levels (105563 posts) is comparable to that obtained from SpaceEx (LGG) after 105563 posts, and XSpeed computes the reachable region $900\times$ faster, as shown in the Table.

\begin{figure}[!htb] 
   \includegraphics[width=\textwidth]{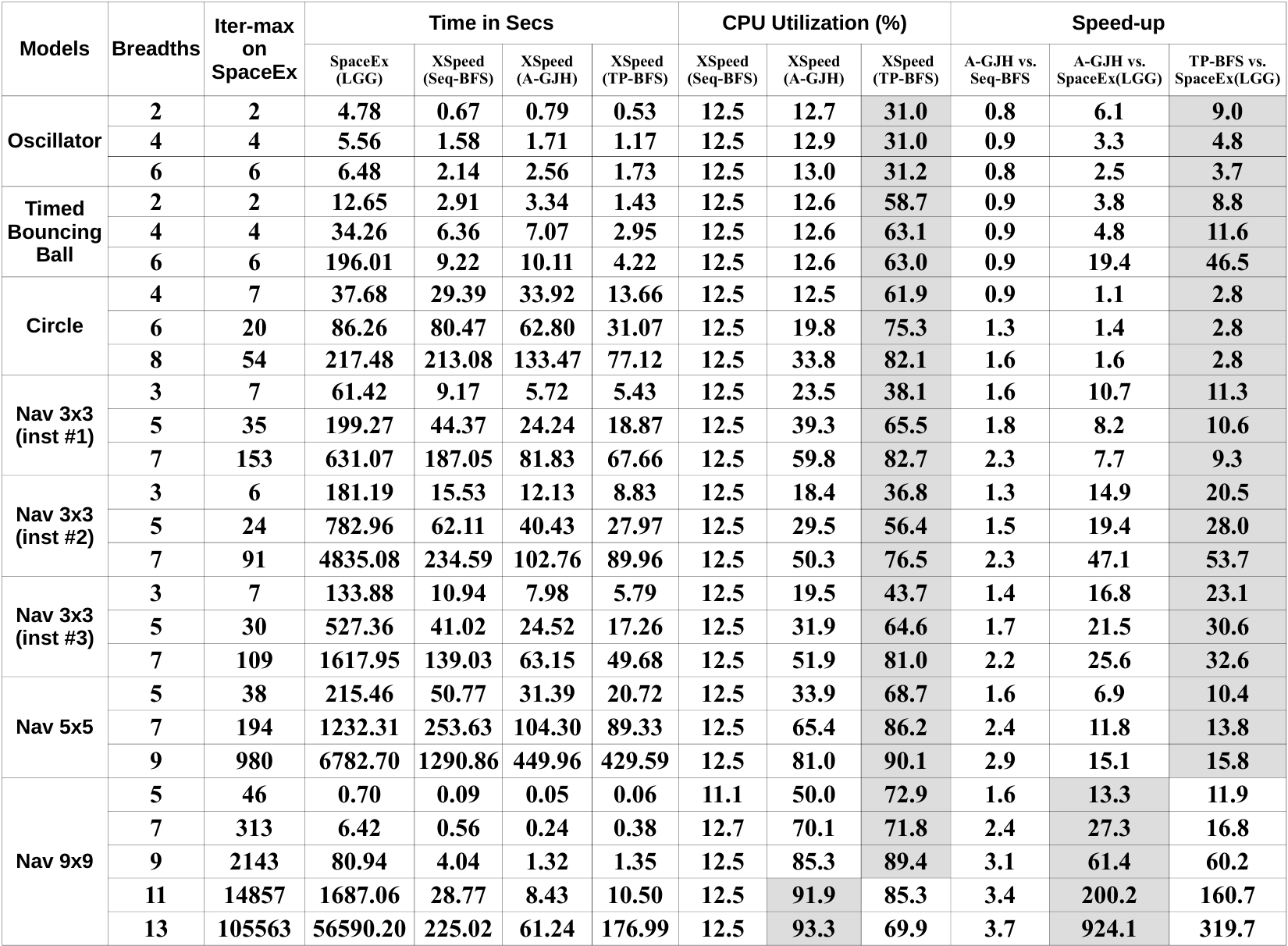}
\caption{Performance comparison of SpaceEx (LGG), sequential BFS, A-GJH and TP-BFS on hybrid systems benchmarks.}
\label{fig:experiments}
\end{figure}

The results for a Circle model is a good illustration of the effectiveness of the TP-BFS algorithm. BFS generates only two new symbolic states at every breadth, one of which exits the location invariants early leaving only one expensive flowpipe to be computed at each level. The A-GJH algorithm is slower than Seq-BFS algorithm due to the parallelization overhead. However, in case of TP-BFS algorithm the flowpipe tasks are distributed across all the available cores, making it faster than the Seq-BFS, A-GJH and SpaceEx (LGG).

We observe that when the number of explored symbolic states (shown as iter-max in Fig. \ref{fig:experiments}) is low to moderate, TP-BFS shows better performance and CPU core utilization in comparison to A-GJH and SpaceEx (LGG). A maximum of $47.1\times$ and  $53.7\times$ is observed on a $3 \times 3$ Navigation benchmark when $7$ BFS levels are explored with a total of $91$ symbolic states using A-GJH and TP-BFS respectively, with respect to SpaceEx's LGG scenario. We observe that when there is a large number of symbolic states in the waiting list, as in the NAV $9 \times 9$ instance, the CPU core utilization and performance of A-GJH is better than TP-BFS. We believe that this is because A-GJH keeps all available cores occupied, even if flowpipe computations are randomly assigned to cores, without taking their cost into consideration. In such a case, the extra overhead with task based load division becomes unnecessary as well as too expensive. This reduces CPU core utilization (since flowpipe cost computation and load-division is a sequential routine) and performance. This is illustrated in Fig. \ref{fig:CPUutilization} which shows that the overhead of load balancing degrades the performance in TP-BFS with the increase in the explored symbolic states, whereas the A-GJH algorithm consistently gains in performance and utilization. We verified that for the considered NAV $9 \times 9$ instance, the waiting list size in the BFS iterations are much larger than the available cores of the processor.

\begin{figure}[!htb] 
\centering{}
   \includegraphics[scale=0.35]{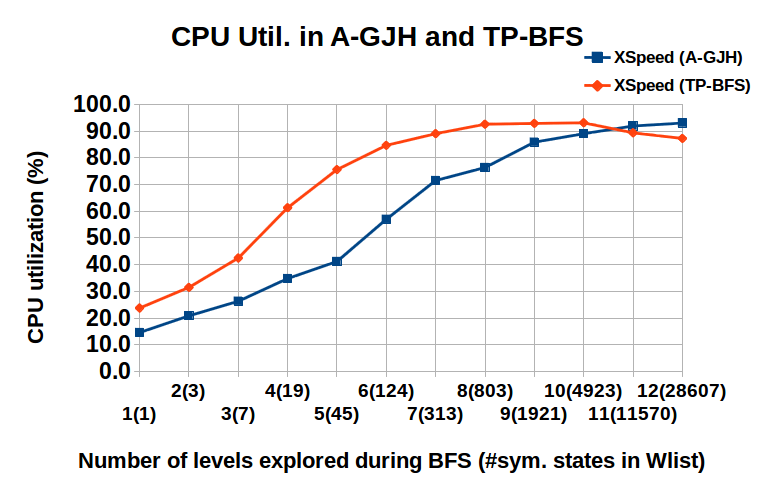}
\caption{Comparison of CPU Utilization between A-GJH and TP-BFS algorithm on a 9 $\times$ 9 Navigation model with sampling time as $1e-3$.}
\label{fig:CPUutilization}
\end{figure}



\section{Conclusion}\label{sec:conclusion}
We present an adaption of G.J. Holzmann's breadth first exploration algorithm of the SPIN model checker, for reachability analysis of hybrid systems. We show that due to the intricacies of post operators in hybrid systems, this first approach does not always produce an efficient load balancing in the hybrid systems scenario. We then propose an alternative approach for load balancing that splits the tasks and distributes them evenly according to an efficiently precomputed cost of the post operations. We provide an implementation of this approach using a support-function based algorithm. Our experiments show that this second approach shows in general a better load-balancing and performance with respect to the first one, with the exception when the number of symbolic states to be explored in the next step is considerably very large. Overall, the two proposed algorithms show a considerable improvement in performance with respect to the current state of the art in reachability analysis for hybrid systems.

\bibliography{mybibliography}

\end{document}